\documentclass{amsart}
\usepackage{amsfonts,amssymb,amscd,amsmath,enumerate,verbatim,calc}
\newtheorem{theorem}{Theorem}[section]
\newtheorem{lemma}[theorem]{Lemma}
\newtheorem{proposition}[theorem]{Proposition}

\newtheorem{definition}[theorem]{Definition}

\newtheorem{example}[theorem]{Example}

\newcommand{\T}{\mathrm}

\newcommand{\Q}{\mathbb{Q}}
\newcommand{\C}{\mathbb{C}}
\newcommand{\R}{\mathbb{R}}
\newcommand{\Z}{\mathbb{Z}}

\begin{document}
\author[M.Shi, Z.Sepasdar, R. S. Wu and P.Sol\'e]
{ MinJia ~ Shi$^{1,*}$, Zahra ~Sepasdar$^{2}$, Patrick ~ Sol\'e$^{3,4}$}

\title[$\Z_{p^k}$-codes]
{Two weight $\Z_{p^k}$-codes, $p$ odd prime}
\subjclass[2010]{94B05, 05C50, 05E30, 11T71} \keywords{Two weight code, Generalized Gray map, Strongly regular graph.}
\thanks{$^*$Corresponding author}
\thanks{E-mail addresses: smjwcl.good@163.com, zahra.sepasdar@gmail.com, patrick.sole@telecom-paristech.fr}
\maketitle

\begin{center}
{\it
${^1}$Key Laboratory of Intelligent Computing \& Signal Processing,
Ministry of Education, Anhui University
No. 3 Feixi Road, Hefei
Anhui Province 230039, P. R. China, National Mobile Communications Research Laboratory,\\ Southeast University, 210096, Nanjing,  P. R. China
and School of Mathematical Sciences of Anhui University, Hefei, 230601, P. R. China
\\${^2}$Department of Pure Mathematics, Ferdowsi University of Mashhad, P. O. Box 1159-91775, Mashhad, Iran\\
   ${^3}$CNRS/LTCI, Telecom ParisTech, University of Paris-Saclay, 75 013 Paris, France\\
$^4$ Department of Mathematics, Faculty of Science, King Abdulaziz University, Jeddah 21589, Saudi Arabia }
\end{center}
\vspace{0.4cm}

\begin{abstract}
We show that regular homogeneous two-weight $\Z_{p^k}$-codes where $p$ is odd and $k\geqslant 2$ with dual Hamming distance at least four do not exist. The proof relies on existence conditions for the strongly regular graph built on the cosets of the dual code.
\end{abstract}

\maketitle

\section{Introduction}
The connection between two-weight codes and strongly regular graphs (SRGs) has been known since the seventies \cite{D}. In that landmark paper, a partial result on the values of the weights of such codes was derived. This result uses in an essential way a SRG defined on the codewords of the two-weight code. Building on that deep result, Calderbank \cite[Theorem 4.4]{Cald} was able to characterize the weights of the projective binary two-weight codes of dual distance at least four. The ingenious proof used a SRG on the cosets of the dual of the two-weight code. As noted in \cite{Cald}, the two SRGs are Delsarte dual to each other.
Codes over $\Z_{p^2},$ for the homogeneous distance have been studied in \cite{AT}.
More recently, the question of  homogeneous 2-weight codes over rings have received some attention \cite{Byrne1,E}. In particular an analogue of Delsarte weight result was derived in \cite{E}. The hypotheses are somewhat more technical requiring in particular the notions of regular and proper codes.

In the present work, building on the mentioned result in \cite{E} we show that regular two homogeneous weight $\Z_{p^k}$-codes where $p$ is odd prime and $k\geqslant 2$ with dual Hamming distance at least four do not exist. Note that the same result does not hold for $\Z_{2^k}$. In \cite{SE} we characterize the two-weight $\Z_{2^k}$-codes that satisfy the hypotheses.

The material is organized as follows. The next section recalls facts and definitions that we need for the following parts. In Section 3 and 4 we obtain the main result. Finally, Section 5 contains conclusion of the paper.

\section{Background}
\subsection{$\Z_{p^k}$-codes}
A linear code $C$ over the ring $R$ of length $n$ is an $R$-submodule of $R^n$.
\begin{definition}\label{4}
Let $R$ be a finite ring. A weight function $w: R\longrightarrow \Q$ is a {\em homogeneous weight}, if $w(0)=0$ and \\
$\T{(i)}$ if $Rx=Ry$ then $w(x)=w(y)$ for all $x,y \in R$.\\
$\T{(ii)}$ there exists a real number $\gamma$ such that $\Sigma_{y\in Rx} w(y)=\gamma \vert Rx \vert$ for all $x \in R\setminus\{0\}$.\\
\end{definition}
We can extend $w$ to a weight function on $R^n$ in the natural way:
$$w(x_1,\dots,x_n)=\sum_{i=1}^{n}w(x_i).$$

\begin{definition}
Let $C$ have  $\ell \times n$ generator matrix $G=[g_1\vert \dots \vert g_n]$ over $R$. The code $C$ is called:\\
$\T{(i)}$ proper for some weight function $w$ if $w(c)=0$ implies $c=0$ for all $c \in C$.\\
$\T{(ii)}$ regular if $\{x\cdot g_i : x \in R^{\ell}\}=R$ for $i=1,\dots,n$.\\
$\T{(ii)}$ projective if $g_i R\neq g_j R$ for any pair of distinct coordinates $i,j \in \{1,\dots,n\}$.
\end{definition}

A linear $\Z_{p^k}$-code of length $n$ is any $\Z_{p^k}$-submodule of $\Z_{p^k}^n.$ For simplicity we call a linear $\Z_{p^k}$-code, a $\Z_{p^k}$-code. The {\em homogeneous weight} for integer rings was introduced in \cite{C} and is defined on $\Z_{p^k}$ by
\begin{equation*}
w_h(x) = \left\{
\begin{array}{rl}
0  &  \text{if } x = 0,\\
p^{k-1}   & \text{if } 0\neq x \in p^{k-1}\Z_{p^k}, \\
(p-1)p^{k-2}  & \text{otherwise}.
\end{array} \right.
\end{equation*}
This weight can be expressed by a character sum. Recall that the {\em character} of a finite abelian group $G$ is a map $\chi: G\longrightarrow \C^*$ such that $\chi(x+y)=\chi(x)\chi(y)$.
\begin{definition}
Character $\chi$ is called a {\em generating character} if every character is of the form $x\mapsto \chi(ax)$ for some $a \in G$.
\end{definition}
It is well-known that for the ring $\Z_n$ the generating character is defined by $\chi(a)=\omega^a$, where $\omega$ is a primitive complex $n^{th}$ root of unity. The next theorem describes the homogeneous weight in terms of characters.
\begin{theorem}[\cite{T}]\label{6}
Let $R$ be a finite Frobenius ring with the generating character $\chi$ and set of units $R^\times.$ Then the homogeneous weights on $R$ are precisely
the functions $w: R \longrightarrow \R $ with $$x \mapsto \gamma \Big{(}1-\frac{1}{\vert R^{\times} \vert} \sum _{u \in R^\times}\chi(xu)\Big{)},$$ where $\gamma$ is like in Definition 2.1.
\end{theorem}

It is easy to see that $\Z_{p^k}^{\times}=\{x \in \Z_{p^k} \mid (x, p^k)=1 \},$
so $\mid \Z_{p^k}^{\times} \mid =\varphi(p^k)=p^k-p^{k-1}$ where $\varphi$ is the Euler's phi function. By Definition \ref{4} (ii), we can take $\gamma=(p-1)p^{k-2}$ for $R=\Z_{p^k}$, thus by Theorem \ref{6}, the following lemma is immediate.
\begin{lemma}[\cite{T}]\label{8}
For any $x \in \Z_{p^k}$, we have $$w_h(x)=(p-1)p^{k-2}-\frac{1}{p} \sum _{u \in \Z_{p^k}^\times}\chi(xu).$$
\end{lemma}
\begin{proposition}[\cite{E}, Corollary 16]\label{5}
Let $C$ be a proper, regular, projective two-weight code with nonzero positive integer weights $w_1< w_2$.
Then there exists a positive integer $d$, a divisor of $\vert C \vert$, and positive integer $t$ such that $w_1=dt$ and $w_2=d(t+1)$.
\end{proposition}
\begin{theorem} \label{byrne}
If $C$ is a regular projective homogeneous two weight $\Z_{p^k}$-code with weights $w_1<w_2,$ then there are integers $u$ and $t$ such that
\begin{eqnarray*}
w_1&=&u p^t\\
w_2&=&(u+1) p^t.
\end{eqnarray*}
\end{theorem}
\begin{proof}
It is a special case of Proposition \ref{5}. Note that the homogeneous weight is positive definite, a fact which makes every $\Z_{p^k}$-code proper.
\end{proof}
The following lemma is essential for Section 3.
\begin{lemma}[\cite{SAN}, Lemma 4.5.4 ]\label{7}\label{element}
Suppose that $C$ is a linear code and $H$ is the parity check matrix for $C$. An vector $v$ is in $C$ iff $vH^T=0$.
\end{lemma}

\subsection{Strongly regular graphs}
A simple graph of order $v$ is called a {\em strongly regular graph} with parameters $(v, \eta, \lambda, \mu)$ whenever it is not complete or edgeless and \\
(i) each vertex is adjacent to $\eta$ vertices,\\
(ii) for each pair of adjacent vertices, there are $\lambda$ vertices adjacent to both,\\
(iii) for each pair of non-adjacent vertices, there are $\mu$ vertices adjacent to both.

An {\em eigenvalue} of a graph, is any eigenvalue of its adjacency matrix. We will call an eigenvalue of $\Gamma$  {\em restricted} if it has an eigenvector which is not a multiple of the all ones vector \textbf{1}. Note that for a $k$-regular connected graph, the restricted eigenvalues are simply the eigenvalues different from $k$.
\begin{theorem}[\cite{B}, Theorem 9.1.2]\label{1}
For a simple graph $\Gamma$, not complete or edgeless, with adjacency matrix $A$, the following are equivalent:\\
$\T{(i)}$ $\Gamma$ is a strongly regular graph.\\
$\T{(ii)}$ $A$ has precisely two distinct restricted eigenvalues.
\end{theorem}
\begin{theorem}[\cite{B}, Theorem 9.1.3]\label{3}
Let $\Gamma$ be a strongly regular graph with adjacency matrix $A$ and parameters $(v, \eta, \lambda, \mu)$. Let $r>s$ be the restricted eigenvalues of $A$. Then:
\begin{equation}
rs=\mu-\eta,
\end{equation}
and
\begin{equation}
r+s=\lambda-\mu.
\end{equation}
By these two relations, it is easy to show that:
\begin{equation}
(r-s)^2=(\lambda-\mu)^2+4(\eta-\mu).
\end{equation}
\end{theorem}
\begin{definition}\label{2}
Let $T$ be a finite group and $S\subseteq T$ be a subset. The corresponding {\em Cayley graph} $C(T,S)$ has vertex set equal to $T$, and two vertices
$g, h \in T$ are adjacent iff $g-h\in S$. The graph is connected iff $S$ generates $T$. Also if $0_T \notin S$ and $-S\subseteq S$ the graph will be loopless and undirected. The graph $C(T,S)$ is a regular graph with degree $\vert S \vert$.
\end{definition}

\section{Syndrome graph}
Suppose that $C$ is a $\Z_{p^k}$-code of length $n$. We define the graph $\Gamma(C)$ as a graph whose vertices are the syndromes
of $C$ and two syndromes are adjacent if they differ by $u$ times a column of the parity check matrix of $C$, where $u \in \Z_{p^k}^\times$:
$$xH^T\sim yH^T \Longleftrightarrow (x-y)H^T=uh_i,$$
where $H=[h_1\mid \dots\mid h_n]$ is the parity check matrix of $C$ and $i=1,\dots,n$.
This graph is a Cayley graph with generator set $S=\{uh_i: u \in \Z_{p^k}^\times, \T{and} \ i=1,\dots,n\}$. As $0 \notin S$ and $-S\subseteq S$  this graph is simple. The graph $\Gamma(C)$ is regular with degree $\vert S \vert=p^{k-1}(p-1)n$ on $\frac{p^{kn}}{\vert C \vert}$ vertices.

It is well-known that there is a one-to-one correspondence between syndromes of a code and its cosets \cite{SAN},
thus this graph can be defined on the cosets of $C$. Suppose that $xH^T \sim yH^T$ in $\Gamma(C)$. Hence
$(x-y)H^T=uh_i$ for some $u \in \Z_{p^k}^\times$ and $i=1,\dots,n$. Consider $a=(a_1,\dots,a_n)$ such that $a_i=u$ and for $j\neq i$, $a_j=0$. So $(x-y)H^T=aH^T$. Therefore two cosets $x+C$ and $y+C$ are adjacent if and only if their difference is a coset $a+C$,
where $a=(a_1,\dots,a_n)$ satisfies $a_i=u$ and $a_j=0$ for $j\neq i$. In this form, the generator set is
$$S=\{a+C: \exists \ i \ \T{such \ that} \ a_i=u\in \Z_{p^k}^\times\  \T{and \ for} \ j\neq i, a_j=0\}.$$

\begin{theorem}\label{10}
Suppose that $C$ is a $\Z_{p^k}$-code with dual homogeneous weights $w_i$ with respective multiplicity $m_i$. Then the eigenvalues of $\Gamma(C)$ are $n(p-1)p^{k-1}-pw_i$ with multiplicity $m_i$.
\end{theorem}
\begin{proof}
For any $x \in C^{\perp}$, we construct the eigenvector $e_x$ of the adjacency matrix $A$ of $\Gamma(C)$ by $(e_x)_y=\chi_x(y)$ where $x=(x_1,\dots,x_n)$
and
$$\chi_x:\frac{\Z_{p^k}^n}{C}\longrightarrow \C $$ $$\chi_x((a_1,\dots,a_n)+C)=\omega ^{(x_1,\dots,x_n)(a_1,\dots,a_n)}.$$
Since $A\cdot(e_x)_y=\Big{(}\sum_{s \in S} \chi_x(s) \Big{)} (e_x)_y$,  the corresponding eigenvalue is $\sum_{s \in S} \chi_x(s)$ (\cite{B}, Page 11). Now, by using Lemma \ref{8}, we have:
\begin{eqnarray*}
w_h(x)&=&\sum_{i=1}^{n}w_h(x_i)=\sum_{i=1}^{n}\Big{(}(p-1)p^{k-2}-\frac{1}{p} \sum _{u \in \Z_{p^k}^\times}\chi(x_iu)\Big{)}\\
      &=&n(p-1)p^{k-2}-\frac{1}{p}\sum_{s \in S} \chi(xs),
\end{eqnarray*}
thus
$\sum_{s \in S} \chi(xs)=n(p-1)p^{k-1}-pw_h(x)$.
\end{proof}

\begin{theorem}\label{11}
Suppose that $C$ is a two-weight code with weights $w_1$ and $w_2$. The coset graph $\Gamma(C^{\perp})$ of  $C^{\perp}$ is a strongly regular graph with degree  $p^{k-1}(p-1)n$ and eigenvalues $K(w_1)$ and $K(w_2)$, where $K(x)=n(p-1)p^{k-1}-px$.
\end{theorem}

\begin{proof}
In the proof of Theorem \ref{10}, for each weight of the dual code we get an eigenvalue for the code in the form of $n(p-1)p^{k-1}-pw_i$. Since $C$ is a two-weight code, the graph $\Gamma(C^{\perp})$ has exactly two eigenvalues $K(w_i)=n(p-1)p^{k-1}-pw_i$ for $i=1,2$. Now, by applying Theorem \ref{1},
this graph is a strongly regular graph.
\end{proof}

\section{Main Resultes}
We need the following lemma to determine the value of $\lambda$.
\begin{lemma} \label{21}
Every element of $\Z_{p^k}^\times$ can be written in $\frac{1}{2}(p^k-2p^{k-1}+1)$ ways as a sum of two other elements in $\Z_{p^k}^\times$.
\end{lemma}
\begin{proof}
It is clear that any element $t \in \Z_{p^k}^\times$ can be written as $$t=(p^k-\ell)+(\ell+t),$$ where $\ell=0,\dots,p^k-1$.
Now, we want to count the number of $\ell$'s such that two elements $(p^k-\ell)$ and $(\ell+t)$ are in $\Z_{p^k}^\times$. The element $p^k-\ell$ is an element of
$\Z_{p^k}^\times$ if $p\nmid \ell$ or in other words $(p^k,\ell)=1$. The number of $\ell$'s such that $(p^k,\ell)=1$ is equal to $p^k-p^{k-1}$.

Note that some of $\ell$ are in $\Z_{p^k}^\times$ but $\ell+t \notin \Z_{p^k}^\times$; we want to remove these elements. The number of $\ell$ such that $\ell$ is in $\Z_{p^k}^\times$ but $\ell+t$ not in $\Z_{p^k}^\times$, is equal to $\mid p\Z_{p^k} \mid =p^{k-1}$.
So there exist $(p^k-2p^{k-1})$ numbers of $\ell$'s such that $(p^k-\ell)$ and $(\ell+t)$ are in $\Z_{p^k}^\times$.

But some terms are counted twice, because we can exchange the roles of $p^k-\ell$ and $\ell+t$, thus the number of $\ell$ is equal to
$(p^k-2p^{k-1})-\frac{1}{2}(p^k-2p^{k-1}-1)= \frac{1}{2}(p^k-2p^{k-1}+1).$
\end{proof}

\begin{example}
For $\Z_9=\Z_{3^2}$ and $\lambda=2$, then the equations are:
\begin{align*}
&1=(9-8)+(8+1)=1+9,  \ \ \ \ \ \    1=(9-7)+(7+1)=2+8, \\ &   1=(9-6)+(6+1)=3+7,  \ \ \ \ \ \    1=(9-5)+(5+1)=4+6,  \\ &
1=(9-4)+(4+1)=5+5,      \ \ \ \ \ \ 1=(9-3)+(3+1)=6+4,   \\ &  1=(9-2)+(2+1)=7+3,  \ \ \ \ \ \   1=(9-1)+(1+1)=8+2, \\ & 1=(9-0)+(0+1)=9+1.
\end{align*}
\end{example}

\begin{theorem} \label{20}
Suppose that $C$ is a two-weight code. If the Hamming distance of $C^{\perp}$ is at least $4,$ then the coset graph $\Gamma(C^{\perp})$ is a SRG with $\lambda=\frac{1}{2}(p^k-2p^{k-1}+1)$.
\end{theorem}
\begin{proof}
Assume that $S$ is the generator of the Cayley graph $\Gamma(C^{\perp})$.  This graph is SRG by Theorem \ref{11}. In $\Gamma(C^{\perp})$, arbitrary vertices $aG^T, bG^T$ and $cG^T$ form a triangle if  $(a-c)G^T, (a-b)G^T \T{and} \ (b-c)G^T$ are elements of $S$.
So there exist some scalars $u_t \in \Z_9^\times$ ($t=1,2,3$) such that $$(a-c)G^T=u_1g_i, \ \ \ (a-b)G^T=u_2g_j, \ \ \T{and} \ \ (b-c)G^T=u_3g_k,$$
for $i,j,k \in \{1,\dots,n \}$.
Note that
$$(a-c)G^T=(a-b)G^T+(b-c)G^T.$$
The hypothesis on the minimum distance of $C^{\perp}$ implies that the only possibility to have such a triangle is that $g_i=g_j=g_k=g$ (otherwise the element $(x_1,\dots,x_n)$ with $$x_i=u_1,\ x_j=-u_2,\ x_k=-u_3 \ {\rm and} \ x_m=0, \ {\rm where} \ m\neq i, j, k$$ is an element of $C^{\perp}$ with $d(x)=3<d(C^{\perp})$). Therefore $u_1g=u_2g+u_3g$ and since  $u_t \in \Z_{p^k}^\times$, we get $u_1=u_2+u_3$.
By Lemma \ref{21}, $u_1$ can be written in $\frac{1}{2}(p^k-2p^{k-1}+1)$ ways as a sum of two other elements.

Thus for arbitrary adjacent vertices $aG^T$ and $cG^T$, there are $\lambda=\frac{1}{2}(p^k-2p^{k-1}+1)$ vetices adjacent to both.
\end{proof}

Now, we can state and prove the main result of this paper.
\begin{theorem}
There exists no  homogeneous two-weight $\Z_{p^k}$-code of dual Hamming distance at least $4$, for $p$ an odd prime and $k\geq 2$.
\end{theorem}

\begin{proof}
Suppose that $C$ is a homogeneous two-weight $\Z_{p^k}$-code.
By Theorem \ref{11}, consider $r_i=n(p-1)p^{k-1}-pw_i$ ($i=1,2$) as the eigenvalues of $\Gamma(C^{\perp})$.  If we set $N=(p-1)p^{k-1}n$ and apply relations (2) and (3) we obtain:
\begin{equation}
\lambda-\mu=2N-p(w_1+w_2),
\end{equation}
\begin{equation}
p^2(w_2-w_1)^2=(\lambda-\mu)^2+4(N-\mu).
\end{equation}
Since $k \geqslant 2$, by Equation (4), $p \mid \lambda-\mu$. Now, we obtain the value of $N$ from (4) and put this value in (5), we get the following equation:
\begin{equation}
p^2(w_2-w_1)^2=(\lambda-\mu)^2+2(\lambda-\mu)+2p(w_2+w_1)-4\mu.
\end{equation}
The equation (6) yields $p \mid 4\mu$ and since $p$ is odd,  $p \mid \mu$. Therefore $p \mid \lambda$. But the condition on the dual distance implies, by Theorem \ref{20}, that $p \nmid \lambda$, which is a contradiction.
\end{proof}

\section{Conclusion}
In this article we have shown the non-existence of regular $\Z_{p^k}$-codes with two nonzero homogeneous weights, and dual distance at least $4$, where $p$ is an odd prime and $k\geq 2$.
The relations between the parameters of a strongly regular graph built on the cosets of the dual code play an essential role in the proof.

\end{document}